\documentclass[11pt,a4paper]{article}
\usepackage[margin=1in]{geometry}
\usepackage[utf8]{inputenc}
\usepackage[T1]{fontenc}
\usepackage[tbtags]{mathtools}
\allowdisplaybreaks
\usepackage{amssymb,amsthm}
\usepackage{multirow}
\usepackage[ttscale=.875]{libertine}
\usepackage[libertine,vvarbb]{newtxmath}
\usepackage{bm}
\usepackage{booktabs}
\usepackage[font=footnotesize]{caption}
\usepackage[svgnames]{xcolor}
\usepackage{hyperref}
\hypersetup{colorlinks={true},linkcolor=[named]{DarkGreen},citecolor={DarkBlue}}
\usepackage[capitalise,nameinlink]{cleveref}

\newtheorem{theorem}{Theorem}

\newtheorem{lemma}{Lemma}

\theoremstyle{definition}
\newtheorem{definition}{Definition}

\def \R{\mathbb{R}}

\newcommand{\vecc}[1]{\ensuremath{\bm{#1}}}
\newcommand{\alloc}{\negthickspace\negthinspace{^\ast}}
\DeclareMathOperator*{\opt}{OPT}
\DeclareMathOperator*{\makespan}{makespan}

\title{A New Lower Bound for Deterministic Truthful Scheduling\thanks{Supported by the Alexander von Humboldt Foundation with funds from the German Federal Ministry of Education and Research (BMBF). 
Yiannis Giannakopoulos is an associated researcher with the Research Training Group GRK 2201 ``Advanced Optimization in a Networked Economy'', funded by the German Research Foundation (DFG).}
}

\author{
		Yiannis Giannakopoulos\thanks{TU Munich. 
		Emails:
		{\tt
		\{\href{mailto:yiannis.giannakopoulos@tum.de}{\nolinkurl{yiannis.giannakopoulos}},
		\href{mailto:alexander.hammerl@tum.de}{\nolinkurl{alexander.hammerl}},
		\href{mailto:diogo.pocas@tum.de}{\nolinkurl{diogo.pocas}}\}@tum.de 
		}}
	\and
		Alexander Hammerl\footnotemark[2]
	\and
		Diogo Poças\footnotemark[2]
}

\date{July 7, 2020}

\begin{document}

\maketitle

\begin{abstract}
We study the problem of truthfully scheduling $m$ tasks to $n$ selfish unrelated
machines, under the objective of makespan minimization, as was introduced in the
seminal work of Nisan and Ronen [STOC'99]. Closing the current gap of $[2.618,n]$ on
the approximation ratio of deterministic truthful mechanisms is a notorious open
problem in the field of algorithmic mechanism design. We provide the first such
improvement in more than a decade, since the lower bounds of $2.414$ (for $n=3$) and
$2.618$ (for $n\to\infty$) by Christodoulou et al. [SODA'07] and Koutsoupias and
Vidali [MFCS'07], respectively. More specifically, we show that the currently best
lower bound of $2.618$ can be achieved even for just $n=4$ machines; for $n=5$ we
already get the first improvement, namely $2.711$; and allowing the number of
machines to grow arbitrarily large we can get a lower bound of $2.755$.
\end{abstract}

\section{Introduction}\label{sec:intro} Truthful scheduling of unrelated parallel
machines is a prototypical problem in \emph{algorithmic mechanism design},
introduced in the seminal paper of Nisan and Ronen~\cite{Nisan1999a} that
essentially initiated this field of research. It is an extension of the classical
combinatorial problem for the makespan minimization objective (see, e.g.,
\cite[Ch.~17]{Vazirani2001a} or~\cite[Sec.~1.4]{HochbaumHall97}), with the added
twist that now machines are rational, \emph{strategic} agents that would not
hesitate to \emph{lie} about their actual processing times for each job, if this can
reduce their personal cost, i.e., their own completion time. The goal is to design a
scheduling mechanism, using payments as incentives for the machines to
\emph{truthfully} report their true processing costs, that allocates all jobs in
order to minimize the makespan, i.e., the maximum completion time across machines.
 
Nisan and Ronen~\cite{Nisan:2001aa} showed right away that no such truthful
deterministic mechanism can achieve an approximation better than $2$ to the optimum
makespan; this is true even for just $n=2$ machines. It is worth emphasizing that
this lower bound is not conditioned on \emph{any} computational complexity
assumptions; it is purely a direct consequence of the added truthfulness requirement
and holds even for mechanisms that have unbounded computational capabilities. It is
interesting to compare this with the classical (i.e., non-strategic) algorithmic
setting where we do know~\cite{Lenstra1990a} that a $2$-approximate polynomial-time
algorithm does exist and that it is NP-hard to approximate the minimum makespan
within a factor smaller than $\frac{3}{2}$. On the positive side, it is also shown
in~\cite{Nisan:2001aa} that the mechanism that myopically allocates each job to the
machine with the fastest reported time for it, and compensates her with a payment
equal to the report of the second-fastest machine, achieves an approximation ratio
of $n$ (where $n$ is the number of machines); this mechanism is truthful and
corresponds to the paradigmatic VCG mechanism (see, e.g., \cite{Nisan2007a}).

Based on these, Nisan and Ronen~\cite[Conjecture~4.9]{Nisan:2001aa} made the bold
conjecture that their upper bound of $n$ is actually the \emph{tight} answer to the
approximation ratio of deterministic scheduling; more than 20 years after the first
conference version of their paper~\cite{Nisan1999a} though, very little progress has
been made in closing their gap of $[2,n]$. Thus, the \emph{Nisan-Ronen conjecture}
remains up to this day one of the most important open questions in algorithmic
mechanism design. Christodoulou et al.~\cite{Christodoulou2007b} improved the lower
bound to $1+\sqrt{2}\approx 2.414$, even for instances with only $n=3$ machines and,
soon after, Koutsoupias and Vidali~\cite{KV07} showed that by allowing $n\to\infty$
the lower bound can be increased to $1+\phi\approx 2.618$. The journal versions of
these papers can be found at~\cite{CKV09} and~\cite{KV13}, respectively. In our
paper we provide the first improvement on this lower bound in well over a decade.

Another line of work tries to provide better lower bounds by imposing further
assumptions on the mechanism, in addition to truthfulness. Most notably, Ashlagi et
al.~\cite{ADL12} were actually able to resolve the Nisan-Ronen conjecture for the
important special case of \emph{anonymous} mechanisms, by providing a lower bound of
$n$. The same can be shown for mechanisms with strongly-monotone allocation
rules~\cite[Sec.~3.2]{Mualem:2018aa} and for mechanisms with additive or local
payment rules~\cite[Sec.~4.3.3]{Nisan:2001aa}.

Better bounds have also been achieved by modifying the scheduling model itself. For
example, Lavi and Swamy~\cite{LS09} showed that if the processing times of all jobs
can take only two values (``high'' and ``low'') then there exists a $2$-approximate
truthful mechanism; they also give a lower bound of $\frac{11}{10}$. Very recently,
Christodoulou et al.~\cite{Christodoulou2020} showed a lower bound of
$\varOmega(\sqrt{n})$ for a slightly generalized model where the completion times of
machines are allowed to be submodular functions (of the costs of the jobs assigned to
them) instead of additive in the standard setting.

Although in this paper we focus exclusively on deterministic mechanisms,
randomization is also of great interest and has attracted a significant amount of
attention~\cite{Nisan:2001aa,Mualem:2018aa,Y09}, in particular the two-machine
case~\cite{Lu2008a,Lu2008b,L09,Chen2015,Kuryatnikova2019}. The currently best
general lower bound on the approximation ratio of randomized (universally) truthful
mechanisms is $2-\frac{1}{n}$~\cite{Mualem:2018aa}, while the upper one is
$0.837n$~\cite{Lu2008b}. For the more relaxed notion of \emph{truthfulness in
expectation}, the upper bound is $\frac{n+5}{2}$~\cite{Lu2008a}. Related to the
randomized case is also the fractional model, where mechanisms (but also the optimum
makespan itself) are allowed to split jobs among machines. For this case,
\cite{Christodoulou2010b} prove lower and upper bounds of $2-\frac{1}{n}$ and
$\frac{n+1}{2}$, respectively; the latter is also shown to be tight for
task-independent mechanisms.

Other variants of the strategic unrelated machine scheduling problem that have been
studied include the Bayesian model~\cite{Chawla2013a,Daskalakis2015b,gkyr2015}
(where job costs are drawn from probability distributions), scheduling without
payments~\cite{Koutsoupias2014a,GKK2016journal} or with
verification~\cite{Nisan:2001aa,Penna2014,Ventre2014}, and strategic behaviour
beyond (dominant-strategy) truthfulness~\cite{fgl2019}. The \emph{related} machines
model, which is essentially a single-dimensional mechanism design variant of our
problem, has of course also been well-studied (see, e.g.,
\cite{Archer2001a,DDDR11,Auletta2009}) and a deterministic PTAS
exists~\cite{Christodoulou2013a}.

\subsection{Our Results and Techniques}\label{sec:results}

We present new lower bounds on the approximation ratio of deterministic truthful
mechanisms for the prototypical problem of scheduling unrelated parallel machines,
under the makespan minimization objective, introduced in the seminal work of Nisan
and Ronen~\cite{Nisan:2001aa}. Our main result (\cref{th:main-lower-bound}) is a
bound of $\rho\approx 2.755$, where $\rho$ is the solution of the cubic
equation~\eqref{eq:lower-bound-infinity}. This improves upon the lower bound of
$1+\phi\approx 2.618$ by Koutsoupias and Vidali~\cite{KV13} which appeared well over
a decade ago~\cite{KV07}. Similar to~\cite{KV13}, we use a family of instances with the number
of machines growing arbitrarily large ($n\to\infty$).

Furthermore, our construction (see~\cref{sec:opt-forumlation}) provides improved
lower bounds also \emph{pointwise}, as a function of the number of machines $n$ that
we are allowed to use. More specifically, for $n=3$ we recover the bound of
$1+\sqrt{2}\approx 2.414$ by~\cite{CKV09}. For $n=4$ we can already match the
$2.618$ bound that~\cite{KV13} could achieve only in the limit as $n\to\infty$. The
first strict improvement, namely $2.711$, comes from $n=5$. As the number of
machines grows, our bound converges to $2.755$. Our results are summarized
in~\cref{table:lower-bounds}.

A central feature of our approach is the formulation of our lower bound as the
solution to a (non-linear) optimization programme~\eqref{eq:nonlinopt}; we then
provide optimal, analytic solutions to it for all values of $n\geq 3$
(\cref{lem:opt-solution}). It is important to clarify here that, in principle, just
giving \emph{feasible} solutions to this programme would still suffice to provide
valid lower bounds for our problem. However, the fact that we pin down and use the
actual \emph{optimal} ones gives rise to an interesting implication: our lower
bounds are provably the best ones that can be derived using our construction.

There are two key elements that allow us to derive our improved bounds, compared to
the approach in previous related works~\cite{CKV09,KV13}. First, we deploy the
weak-monotonicity (\cref{th:wmon}) characterization of truthfulness in a slightly
more delicate way; see~\cref{lem:stv}. This gives us better control and flexibility
in considering deviating strategies for the machines (see our case-analysis
in~\cref{sec:lower-bound}). Secondly, we consider more involved instances, with two
auxiliary parameters (namely $r$ and $a$; see, e.g., \eqref{eq:costmatrixa1}
and~\eqref{eq:randa}) instead of just one. On the one hand, this increases the
complexity of the solution, which now has to be expressed in an implicit way via the
aforementioned optimization programme~\eqref{eq:nonlinopt}. But at the same time,
fine-tuning the optimal choice of the variables allows us to (provably) push our
technique to its limits. Finally, let us mention that, for a small number of
machines ($n=3,4,5$) we get $r=1/a$ in an optimal choice of parameters. Under
$r=1/a$, we end up with $a$ as the only free parameter, and our construction becomes
closer to that of~\cite{CKV09,KV13}; in fact, for $3$ machines it is essentially the
same construction as in \cite{CKV09} (which explains why we recover the same lower
bound). However, for $n\geq 6$ machines we need a more delicate choice of $r$.

\begin{table}[t]
$$
\begin{array}{l  c c c c c c c c c}  
\toprule
n    & 3 & 4 & 5 & 6 & 7 & 8 & \dots & \infty  \\
\midrule
\text{Previous work}      & 2.414 & 2.465 & 2.534 & 2.570 & 2.590 & 2.601 & \dots & 2.618      \\
\text{This paper}         & 2.414 & 2.618 & 2.711 & 2.739 & 2.746 & 2.750 & \dots & 2.755     \\
\bottomrule
\end{array}
$$
\caption{Lower bounds on the approximation ratio of deterministic truthful scheduling, as a
function of the number of machines $n$, given by our~\cref{th:main-lower-bound}
(bottom line). The previous state-of-the-art is given in the line above and first
appeared in~\cite{Christodoulou2007b} ($n=3$) and \cite{KV07} ($n\geq 4$). The case
with $n=2$ machines was completely resolved in \cite{Nisan1999a}, with an
approximation ratio of $2$.}
\label{table:lower-bounds}
\end{table}

\section{Notation and Preliminaries} \label{sec:prelims}

Before we go into the construction of our lower bound (\cref{sec:lower-bound}), we
use this section to introduce basic notation and recall the notions of mechanism,
truthfulness, monotonicity, and approximation ratio. We also provide a technical
tool (\cref{lem:stv}) that is a consequence of weak monotonicity (\cref{th:wmon});
this lemma will be used several times in the proof of our main result.

\subsection{Unrelated Machine Scheduling}

In the unrelated machine scheduling setting, we have a number $n$ of machines and a
number $m$ of tasks to allocate to these machines. These tasks can be performed in
any order, and each task has to be assigned to exactly one machine; machine $i$
requires $t_{ij}$ units of time to process task $j$. Hence, the complete description
of a problem instance can be given by a $n\times m$ \emph{cost matrix} of the values
$t_{ij}$, which we denote by $\vecc{t}$. In this matrix, row $i$, denoted by
$\vecc{t}_i$, represents the processing times for machine $i$ (on the different
tasks) and column $j$, denoted by $\vecc{t}_j$, represents the processing times for
task $j$ (on the different machines). These values $t_{ij}$ are assumed to be
nonnegative real quantities, $t_{ij}\in\R_{+}$.

Applying the methodology of mechanism design, we assume that the processing times
for machine $i$ are known only by machine $i$ herself. Moreover, machines are
selfish agents; in particular, they are not interested in running a task unless they
receive some compensation for doing so. They may also lie about their processing
times if this would benefit them. This leads us to consider the central notion of
(direct-revelation) \emph{mechanisms}: each machine reports her values, and a
mechanism decides on an allocation of tasks to machines, as well as corresponding
payments, based on the reported values.

\begin{definition}[Allocation rule, payment rule, mechanism]\label{def:mechanism}
Given $n$ machines and $m$ tasks,
\begin{itemize}
	\item a (deterministic) \emph{allocation rule} is a function that describes the
	allocation of tasks to machines for each problem instance. Formally, it is
	represented as a function $\vecc{a}:\R_{+}^{n\times m}\rightarrow
	\{0,1\}^{n\times m}$ such that, for every $\vecc{t}=(t_{ij})\in \R_{+}^{n\times
	m}$ and every task $j=1,\ldots, m$, there is exactly one machine $i$ with
	$a_{ij}(\vecc{t})=1$, that is,
	\begin{equation}\label{eq:alloc}\sum_{i=1}^n a_{ij}(\vecc{t})=1;\end{equation}
	\item a \emph{payment rule} is a function that describes the payments to
	machines for each problem instance. Formally, it is represented as a function
	$p:\R_{+}^{n\times m}\rightarrow\R^n$;
	\item a (direct-revelation, deterministic) \emph{mechanism} is a pair
	$(\vecc{a},p)$ consisting of an allocation and payment rules.
\end{itemize}
\end{definition}

We let $\mathbb{A}$ denote the set of feasible allocations, that is, matrices
$\vecc{a}=(a_{ij})\in\{0,1\}^{n\times m}$ satisfying \eqref{eq:alloc}. Given a
feasible allocation $\vecc{a}$, we let $\vecc{a}_i$ denote its row $i$, that is, the
allocation to machine $i$. Similarly, given a payment vector $\vecc{p}\in\R^n$, we
let $p_i$ denote the payment to machine $i$; note that the payments represent an
amount of money given to the machine, which is somewhat the opposite situation
compared to other mechanism design frameworks (such as auctions, where payments are
done by the agents to the mechanism designer).

\subsection{Truthfulness and Monotonicity}

Whenever a mechanism assigns an allocation $\vecc{a}_i$ and a payment $p_i$ to
machine $i$, this machine incurs a quasi-linear \emph{utility} equal to her payment
minus the sum of processing times of the tasks allocated to her,
\[p_i-\vecc{a}_i\cdot\vecc{t}_i=p_i-\sum_{j=1}^m a_{ij}t_{ij}.\]

Note that the above quantity depends on the machine's both \emph{true} and \emph{reported} processing times,
which in principle might differ. As
already explained, machines behave selfishly. Thus, from the point of view of a
mechanism designer, we wish to ensure a predictable behaviour of all parties
involved. In particular, we are only interested in mechanisms that encourage agents
to report their true valuations.

\begin{definition}[Truthful mechanism]\label{def:truthful}
 A mechanism $(\vecc{a},p)$ is \emph{truthful} if every machine maximizes their
 utility by reporting truthfully, regardless of the reports by the other machines.
 Formally, for every machine $i$, every $\vecc{t}_i,\vecc{t}'_i\in \R_{+}^m$,
 $\vecc{t}_{-i}\in R_{+}^{(n-1)\times m}$, we have that
 \begin{equation}\label{eq:def_truthfulness}p_i(\vecc{t}_i,\vecc{t}_{-i})-\vecc{a}_i(\vecc{t}_i,\vecc{t}_{-i})\cdot\vecc{t}_i\geq
 p_i(\vecc{t}'_i,\vecc{t}_{-i})-\vecc{a}_i(\vecc{t}'_i,\vecc{t}_{-i})\cdot\vecc{t}_i.\tag{TR}\end{equation}
\end{definition}

In \eqref{eq:def_truthfulness}, we ``freeze'' the reports of all machines other than
$i$. The left hand side corresponds to the utility achieved by machine $i$ when her
processing times correspond to $\vecc{t}_i$ and she truthfully reports $\vecc{t}_i$.
The right hand side corresponds to the utility achieved if machine $i$ lies and
reports $\vecc{t}'_i$.

The most important example of a truthful mechanism in this setting is the VCG
mechanism that assigns each task independently to the machine that can perform it
fastest, and paying that machine (for that task) a value equal to the second-lowest
processing time. Note that this is somewhat the equivalent of second-price auctions
(that sell each item independently) for the scheduling setting.

A fundamental result in the theory of mechanism design is the very useful property
of truthful mechanisms, in terms of ``local'' monotonicity of the allocation
function with respect to single-machine deviations.

\begin{theorem}[Weak monotonicity~\cite{Nisan:2001aa,LS09}]
\label{th:wmon} 
Let $\vecc{t}$ be a cost matrix, $i$ be a machine, and $\vecc{t}'_i$
another report from machine $i$. Let $\vecc{a}_i$ be the allocation of $i$ for cost
matrix $\vecc{t}$ and $\vecc{a}'_i$ be the allocation of $i$ for cost matrix
$(\vecc{t}'_i,\vecc{t}_{-i})$. Then, if the mechanism is truthful, it must be that
\begin{equation}\label{eq:truthfulness}
(\vecc{a}_i-\vecc{a}'_i)\cdot(\vecc{t}_i-\vecc{t}'_i)\leq 0.
\tag{WMON}
\end{equation}
\end{theorem}

As a matter of fact, \eqref{eq:truthfulness} is also a \emph{sufficient} condition
for truthfulness, thus providing an exact \emph{characterization} of
truthfulness~\cite{Saks2005a}. However, for our purposes in this paper we will only
need the direction in the statement of~\cref{th:wmon} as stated above. We will make
use of the following lemma, which exploits the notion of weak monotonicity in a
straightforward way. The second part of this lemma can be understood as a refinement
of a technical lemma that appeared before in~\cite[Lemma 2]{CKV09} (see
also~\cite[Lemma 1]{KV13}).

\begin{lemma}
\label{lem:stv}
Suppose that machine $i$ changes her report from $\vecc{t}$ to $\vecc{t}'$, and that
a truthful mechanism correspondingly changes her allocation from $\vecc{a}_i$ to
$\vecc{a}'_i$. Let $\{1,\ldots,m\}=S\cup T\cup V$ be a partition of the tasks into
three disjoint sets.
\begin{enumerate}
\item Suppose that (a) the costs of $i$ on $V$ do not change, that is,
$\vecc{t}_{i,V}=\vecc{t}'_{i,V}$ and (b) the allocation of $i$ on $S$ does not
change, that is, $\vecc{a}_{i,S}=\vecc{a}'_{i,S}$. Then
\[(\vecc{a}_{i,T}-\vecc{a}'_{i,T})\cdot(\vecc{t}_{i,T}-\vecc{t}'_{i,T})\leq 0.\]
\item Suppose additionally that (c) the costs of $i$ strictly decrease on her
allocated tasks in $T$ and strictly increase on her unallocated tasks in $T$. Then
her allocation on $T$ does not change, that is, $\vecc{a}_{i,T}=\vecc{a}'_{i,T}$.
\end{enumerate}
\end{lemma}

\begin{proof}
To prove the first point, simply apply \eqref{eq:truthfulness} and split the sum into the three sets of tasks,
\begin{align*}
0&\geq(\vecc{a}_i-\vecc{a}'_i)\cdot(\vecc{t}_i-\vecc{t}'_i)\\
&=(\vecc{a}_{i,S}-\vecc{a}'_{i,S})\cdot(\vecc{t}_{i,S}-\vecc{t}'_{i,S})+(\vecc{a}_{i,T}-\vecc{a}'_{i,T})\cdot(\vecc{t}_{i,T}-\vecc{t}'_{i,T})+(\vecc{a}_{i,V}-\vecc{a}'_{i,V})\cdot(\vecc{t}_{i,V}-\vecc{t}'_{i,V});\end{align*}
since $\vecc{t}_{i,V}=\vecc{t}'_{i,V}$ and $\vecc{a}_{i,S}=\vecc{a}'_{i,S}$, the
result follows.

To prove the second point, we look at each term appearing in the inner product
$(\vecc{a}_{i,T}-\vecc{a}'_{i,T})\cdot(\vecc{t}_{i,T}-\vecc{t}'_{i,T})$. Let $j\in
T$ be a task which was originally allocated to machine $i$; then, $a_{i,j}=1$ and,
by assumption, $t_{i,j}>t'_{i,j}$. Since $a'_{i,j}$ is either $1$ or $0$, it follows
that $(a_{i,j}-a'_{i,j})(t_{i,j}-t'_{i,j})$ is either $0$ (if the allocation does
not change) or $t_{i,j}-t'_{i,j}>0$ (if the allocation changes). Similarly, assume
now that $j\in T$ was originally not allocated to machine $i$; then, $a_{i,j}=0$
and, by assumption, $t_{i,j}<t'_{i,j}$. Since $a'_{i,j}$ is either $0$ or $1$, it
follows that $(a_{i,j}-a'_{i,j})(t_{i,j}-t'_{i,j})$ is either $0$ (if the allocation
does not change) or $(-1)\cdot(t_{i,j}-t'_{i,j})>0$ (if the allocation changes). By
the first point, the sum over all these terms must be non-positive. We conclude that
all these terms must be zero, and hence, the allocation of machine $i$ for tasks on
$T$ must not change.
\end{proof}

\subsection{Approximation ratio}

One of the main open questions in the theory of algorithmic mechanism design is to figure out
what is the ``best'' possible truthful mechanism, with respect to the objective of makespan minimization. This can be quantified in terms of the approximation ratio of a mechanism.

\begin{definition} Given $n$ machines and $m$ tasks:
\begin{itemize}
	\item Let $\vecc{a}$ be a feasible allocation and $\vecc{t}$ a problem instance. The \emph{makespan} of $\vecc{a}$ on $\vecc{t}$ is defined as the quantity
	\begin{equation*}\makespan(\vecc{a},\vecc{t})=\max_{i=1,\ldots,n}\sum_{j=1}^m a_{ij}t_{ij}.\end{equation*}
	\item Let $\vecc{t}$ be a problem instance. The \emph{optimal makespan} is defined as the quantity
	\begin{equation*}\opt(\vecc{t})=\min_{\vecc{a}\in \mathbb{A}}\makespan(\vecc{a},\vecc{t}).\end{equation*}
	\item Let $\vecc{a}$ be an allocation rule. We say that $\vecc{a}$ has \emph{approximation ratio} $\rho\geq 1$ if, for any problem instance $\vecc{t}$, we have that
	\[\makespan(\vecc{a}(\vecc{t}),\vecc{t})\leq\rho\opt(\vecc{t});\]
	if no such quantity $\rho$ exists, we say that $\vecc{a}$ has \emph{infinite approximation ratio}.
\end{itemize}
\end{definition}

As shown in \cite{Nisan:2001aa}, the VCG mechanism has an approximation ratio
of $n$, the number of machines. The long-standing conjecture by Nisan and Ronen states that this mechanism is essentially the best one; any truthful mechanism
is believed to attain a worst-case approximation ratio of $n$ (for sufficiently many tasks). In this paper, we prove lower bounds on the approximation ratio of any
truthful mechanism (\cref{table:lower-bounds,th:main-lower-bound}); our bounds converge to $2.755$ as $n\rightarrow\infty$.

\section{Lower Bound}
\label{sec:lower-bound}

To prove our lower bound, from here on we assume $n\geq 3$ machines, since the case
$n=1$ is trivial and the case $n=2$ is resolved by \cite{Nisan1999a} (with an
approximation ratio of $2$). Our construction will be made with the choice of two
parameters $r,a$. For now we shall simply assume that $a>1>r>0$. Later we will
optimize the choices of $r$ and $a$ in order to achieve the best lower bound
possible by our construction.

We will use $L_n$ to denote the $n\times n$ matrix with $0$ in its diagonal and
$\infty$ elsewhere,

\[L_n=\left[
\begin{array}{cccc}
0&\infty&\cdots&\infty\\
\infty & 0 & \cdots & \infty\\
\vdots & \vdots &\ddots & \vdots\\
\infty & \infty &\cdots & 0
\end{array}
\right].\]

We should mention here that allowing $t_{ij}=\infty$ is a technical convenience. If
only finite values are allowed, we can replace $\infty$ by an arbitrarily high
value. We also follow the usual convention, and use an asterisk $^\ast$ to denote a
full or partial allocation. Our lower bound begins with the following cost matrix
for $n$ machines and $2n-1$ tasks:

\begin{equation}
A_0=\left[\begin{array}{c|cccccc}
\multirow{7}{*}{\raisebox{-8pt}{$L_n$}}& \alloc 1 & 1 & a^{-1} & a^{-2} & \cdots & a^{-n+3}\\
& 1 & 1 & a^{-1} & a^{-2} & \cdots & a^{-n+3}\\
& \infty & 1 & \infty & \infty & \cdots & \infty\\
& \infty & \infty & a^{-1} &\infty & \cdots & \infty\\
& \infty & \infty & \infty & a^{-2} & \ddots & \infty\\
& \vdots & \vdots & \vdots & \ddots & \ddots & \vdots\\
& \infty & \infty & \infty & \infty & \cdots & a^{-n+3}\end{array}\right].\label{eq:costmatrixa0}
    \end{equation}
The tasks of cost matrix $A_0$ can be partitioned in two groups. The first $n$ tasks
(i.e., the ones corresponding to the $L_n$ submatrix) will be called \emph{dummy}
tasks. Machine $i$ has a cost of $0$ for dummy task $i$ and a cost of $\infty$ for
all other dummy tasks. The second group of tasks, numbered $n+1,\ldots, 2n-1$, will
be called \emph{proper} tasks. Notice that machines $1$ and $2$ have the same costs
for proper tasks; they both need time $1$ to execute task $n+1$ and time $a^{-j+2}$
to execute task $n+j$, for all $j=2,\ldots n-1$. Finally for $i\geq 3$, machine $i$
has a cost of $a^{-i+3}$ on proper task $n+i-1$ and $\infty$ cost for all other
proper tasks.

In order for a mechanism to have a finite approximation ratio, it must not assign
any tasks with unbounded costs. In particular, each dummy task must be assigned to
the unique machine that completes it in time $0$; and proper task $n+1$ must be
assigned to either machine $1$ or $2$. Since the costs of machines $1$ and $2$ are
the same on all proper tasks, we can without loss assume that machine $1$ receives
proper task $n+1$. Hence, the allocation on $A_0$ should be as (designated by an
asterisk) in \eqref{eq:costmatrixa0}.

Next, we reduce the costs of all proper tasks for machine $1$, and get the cost matrix

\begin{equation}
A_1=\left[\begin{array}{c|cccccc}
\multirow{7}{*}{\raisebox{-8pt}{$L_n$}}& r & a^{-1} & a^{-2} & a^{-3} & \cdots & a^{-n+2}\\
& 1 & 1 & a^{-1} & a^{-2} & \cdots & a^{-n+3}\\
& \infty & 1 & \infty & \infty & \cdots & \infty\\
& \infty & \infty & a^{-1} &\infty & \cdots & \infty\\
& \infty & \infty & \infty & a^{-2} & \ddots & \infty\\
& \vdots & \vdots & \vdots & \ddots & \ddots & \vdots\\
& \infty & \infty & \infty & \infty & \cdots & a^{-n+3}\end{array}\right].\label{eq:costmatrixa1}
\end{equation}
Under the new matrix $A_1$, the cost of machine $1$ for proper task $n+1$ is reduced
from $1$ to $r$; and her cost for any other proper task $n+j$, $j=2,\ldots,n-1$, is
reduced by a factor of $a$, that is, from $a^{-j+2}$ to $a^{-j+1}$. The key idea in
this step is the following: we want to impose a constraint on $r$ and $a$ that
ensures that \emph{at least one} of the proper tasks $n+1,n+2$ is still allocated to
machine $1$. Using the properties of truthfulness, this can be achieved via the
following lemma:

\begin{lemma}
\label{lem:randa}
Consider a truthful scheduling mechanism that, on cost matrix $A_0$, assigns proper
task $n+1$ to machine $1$. Suppose also that
\begin{equation}\label{eq:randa}
1-r > a^{-1}-a^{-n+2}.
\end{equation}
Then, on cost matrix $A_1$, machine $1$ must receive at least one of the proper
tasks $n+1,n+2$.
\end{lemma}

\begin{proof}
We apply part 1 of \cref{lem:stv}, taking $S=\emptyset$, $V$ as the set of dummy
tasks, and $T$ as the set of proper tasks. If $\vecc{a}_1$, $\vecc{a}'_1$ denote the
allocations of machine $1$ for cost matrices $A_0$, $A_1$ respectively, we get that
\[(\vecc{a}_{1,T}-\vecc{a}'_{1,T})\cdot(\vecc{t}_{1,T}-\vecc{t}'_{1,T})\leq 0.\]
Assume further, for the sake of obtaining a contradiction, that on cost matrix
$A_1$, machine 1 does not get either task $n+1$ or $n+2$; that is,
$a'_{1,n+1}=a'_{1,n+2}=0$. Notice that $a_{1,n+1}=1$ (since machine 1 gets task
$n+1$ on cost matrix $A_0$) and we have the lower bounds $a_{1,n+2}\geq 0$ as well
as $a_{1,n+j}-a'_{1,n+j}\geq -1$ for $j=3,\ldots,n-1$. Combining all these, we get
\begin{align*}0&\geq(\vecc{a}_{1,T}-\vecc{a}'_{1,T})\cdot(\vecc{t}_{1,T}-\vecc{t}'_{1,T})\\
&=(a_{1,n+1}-a'_{1,n+1})(t_{1,n+1}-t'_{1,n+1})+(a_{1,n+2}-a'_{1,n+2})(t_{1,n+2}-t'_{1,n+2})\\
&\qquad+(a_{1,n+3}-a'_{1,n+3})(t_{1,n+3}-t'_{1,n+3})+\ldots+(a_{1,2n-1}-a'_{1,2n-1})(t_{1,2n-1}-t'_{1,2n-1})\\
&\geq 1\cdot(1-r)+0\cdot(1-a^{-1})+(-1)\cdot(a^{-1}-a^{-2})+\ldots+(-1)\cdot(a^{-n+3}-a^{-n+2})\\
&=1-r-a^{-1}+a^{-n+2},
\end{align*}
where in the last step we observe that the terms for tasks $n+3,\ldots,2n-1$ form a telescoping sum. Thus, we obtain that $1-r\leq a^{-1}-a^{-n+2}$, which contradicts our original assumption \eqref{eq:randa}.
\end{proof}

For the remainder of our construction, we assume that $r$ and $a$ are such that
\eqref{eq:randa} is satisfied. Next, we split the analysis depending on the
allocation of the proper tasks $n+1,\ldots 2n-1$ to machine $1$ on cost matrix $A_1$, as restricted
by~\cref{lem:randa}.

\subsection{Case 1: Machine $1$ gets \emph{all} proper tasks}

In this case, we perform the following changes in machine $1$'s tasks, obtaining a
new cost matrix $B_1$. We increase the cost of dummy task $1$, from $0$ to $1$, and
we decrease the costs of all her proper tasks by an arbitrarily small amount. Notice
that
\begin{itemize}
\item for the mechanism to achieve a finite approximation ratio, it must still
allocate the dummy task 1 to machine 1;
\item given that the mechanism does not change the allocation on dummy task 1, and
that machine 1 only decreases the completion times of her proper tasks, part 2 of
\cref{lem:stv} implies that machine 1 still gets all proper tasks.
\end{itemize}
Thus, the allocation must be as shown below (for ease of exposition, in the cost
matrices that follow we omit the ``arbitrarily small'' amounts by which we change
allocated / unallocated tasks):

\[B_1=\left[\begin{array}{cccccc|ccccc}
\alloc 1 & \infty & \infty & \infty & \cdots & \infty & \alloc r & \alloc a^{-1} & \alloc a^{-2} & \cdots & \alloc a^{-n+2}\\
\infty & \alloc 0 & \infty & \infty & \cdots & \infty & 1 & 1 & a^{-1} & \cdots & a^{-n+3}\\
\infty & \infty & \alloc 0 & \infty & \cdots & \infty &\infty & 1 & \infty & \cdots & \infty\\
\infty & \infty & \infty & \alloc 0 & \cdots & \infty &\infty &\infty & a^{-1} & \cdots & \infty\\
\vdots & \vdots & \vdots & \vdots & \ddots & \vdots & \vdots & \vdots & \vdots & \ddots & \vdots\\
\infty & \infty & \infty & \infty & \cdots & \alloc 0 & \infty & \infty & \infty & \cdots & a^{-n+3}\end{array}\right].\]
This allocation achieves a makespan of $1+r+a^{-1}+\ldots+a^{-n+2}$, while a
makespan of $1$ can be achieved by assigning each proper task $n+j$ to machine
$j+1$. Hence, this case yields an approximation ratio of at least
$1+r+a^{-1}+\ldots+a^{-n+2}$.

\subsection{Case 2: Machine $1$ gets task $n+1$, but does \emph{not} get all proper tasks.}

That is, at least one of tasks $n+2,\ldots 2n-1$ is not assigned to machine $1$.
Suppose that task $n+j$ is the lowest indexed proper task that is not allocated to
her. We decrease the costs of her \emph{allocated} proper tasks $n+1,\ldots,n+j-1$
to $0$, while increasing the cost $a^{-j+1}$ of her (unallocated) proper task $n+j$ by an
arbitrarily small amount. By \cref{lem:stv}, the allocation of machine $1$ on the
proper tasks $n+1,\ldots,n+j$ does not change. Hence we get a cost matrix of the
form

\[B_2=\left[\begin{array}{c|cccccc}
\multirow{7}{*}{\raisebox{-15pt}{$L_n$}}& \alloc 0 & \alloc 0 & \cdots & a^{-j+1} & \cdots & a^{-n+2}\\
& 1 & 1 & \cdots & a^{-j+2} & \cdots & a^{-n+3}\\
& \infty & 1 & \cdots & \infty & \cdots & \infty\\
& \vdots & \vdots & \ddots &\vdots & \cdots & \infty\\
&\infty & \infty & \cdots & a^{-j+2} & \ddots & \infty\\
& \vdots & \vdots & \vdots & \ddots & \ddots & \vdots\\
& \infty & \infty & \infty & \infty & \cdots & a^{-n+3}\end{array}\right].\]

Since task $n+j$ is not allocated to machine $1$, and the mechanism has finite
approximation ratio, it must be allocated to either machine $2$ or machine $j+1$. In
either case, we increase the cost of the dummy task of this machine from $0$ to
$a^{-j+1}$, while decreasing the cost of her proper task $n+j$ by an arbitrarily
small amount. For example, if machine $2$ got task $n+j$, we would end up with

\[C_2=\left[\begin{array}{ccccccc|cccccc}
\alloc 0 & \infty & \infty & \cdots & \infty & \cdots & \infty & 0 & 0 & \cdots & a^{-j+1} & \cdots & a^{-n+2}\\
\infty & \alloc a^{-j+1} & \infty & \cdots & \infty & \cdots & \infty & 1 & 1 & \cdots & \alloc a^{-j+2} & \cdots & a^{-n+3}\\
\infty & \infty & \alloc 0 & \cdots & \infty & \cdots & \infty & \infty & 1 & \cdots & \infty & \cdots & \infty\\
\vdots & \vdots & \vdots & \ddots & \vdots & \ddots & \vdots & \vdots & \vdots & \ddots &\vdots & \ddots & \infty\\
\infty & \infty & \infty & \cdots & \alloc 0 & \cdots & \infty & \infty & \infty & \cdots & a^{-j+2} & \cdots & \infty\\
\vdots & \vdots & \vdots & \ddots & \vdots & \ddots & \vdots & \vdots & \vdots & \ddots & \vdots & \ddots & \vdots\\
\infty & \infty & \infty & \cdots & \infty & \cdots & \alloc 0 & \infty & \infty & \infty & \infty & \cdots & a^{-n+3}\end{array}\right].\]

Similarly to the previous Case~1, the mechanism must still allocate the dummy task
to this machine, and given that the allocation does not change on the dummy task,
\cref{lem:stv} implies that the allocation must also remain unchanged on the proper
task $n+j$. Finally, observe that the present allocation achieves a makespan of at
least $a^{-j+1}+a^{-j+2}$, while a makespan of $a^{-j+1}$ can be achieved by
assigning proper task $n+j$ to machine $1$ and proper task $n+j'$ to machine $j'+1$,
for $j'>j$. Hence, this case yields an approximation ratio of at least
\begin{equation*}
\frac{a^{-j+1}+a^{-j+2}}{a^{-j+1}}=1+a.
\end{equation*}

\subsection{Case 3: Machine 1 does \emph{not} get task $n+1$}

By \cref{lem:randa}, machine $1$ must receive proper task $n+2$. In this case, we
decrease the cost of her task $n+2$, from $a^{-1}$ to $0$, while increasing the
cost $r$ of her (unallocated) task $n+1$ by an arbitrarily small amount.
Since by truthfulness, the allocation of machine $1$ for these two tasks does not change, the allocation
must be as below:

\[B_3=\left[\begin{array}{c|cccccc}
\multirow{6}{*}{\raisebox{-3pt}{$L_n$}}& r & \alloc 0 & a^{-2} & \cdots & a^{-n+2}\\
& \alloc  1 & 1 & a^{-1} & \cdots & a^{-n+3}\\
& \infty & 1 & \infty & \cdots & \infty\\
& \infty & \infty & a^{-1} & \cdots & \infty\\
& \vdots & \vdots & \vdots &\ddots & \vdots\\
& \infty & \infty & \infty & \cdots & a^{-n+3}\end{array}\right].\]

Since task $n+1$ is not allocated to machine $1$, and the mechanism has finite
approximation ratio, it must be allocated to machine $2$. We now increase the cost
of the dummy task of machine $2$ from $0$ to $\max\{r,a^{-1}\}$, while decreasing
the cost of her proper task $n+1$ by an arbitrarily small amount. Similarly to
Cases~1 and 2, the mechanism must still allocate the dummy task to machine $2$, and
preserve the allocation of machine $2$ on the proper task $n+1$. Thus, we get the
allocation shown below:

\[C_3=\left[\begin{array}{cccccc|cccccc}
\alloc 0 & \infty & \infty & \infty & \cdots & \infty & r & 0 & a^{-2} & \cdots & a^{-n+2}\\
\infty & \alloc \max\{r,a^{-1}\} & \infty & \infty & \cdots & \infty & \alloc 1 & 1 & a^{-1} & \cdots & a^{-n+3}\\
\infty & \infty & \alloc 0 & \infty & \cdots & \vdots & \infty & 1 & \infty & \cdots & \infty\\
\infty & \infty & \infty & \alloc 0 & \cdots & \infty & \infty & \infty & a^{-1} & \cdots & \infty\\
\vdots & \vdots & \vdots & \vdots & \ddots & \vdots & \vdots & \vdots & \vdots &\ddots & \vdots\\
\infty & \infty & \infty & \infty & \cdots & \alloc 0 & \infty & \infty & \infty & \cdots & a^{-n+3}\end{array}\right].\]
This allocation achieves a makespan of at least $1+\max\{r,a^{-1}\}$, while a
makespan of $\max\{r,a^{-1}\}$ can be achieved by assigning proper tasks $n+1,n+2$
to machine $1$ and proper task $n+j'$ to machine $j'+1$, for all $j'>2$. Hence, this
case yields an approximation ratio of at least
\[\frac{1+\max\{r,a^{-1}\}}{\max\{r,a^{-1}\}}=1+\min\{r^{-1},a\}.\]

\subsection{Main result}
\label{sec:opt-forumlation}

The three cases considered above give rise to possibly different approximation
ratios; our construction will then yield a lower bound equal to the \emph{smallest}
of these ratios. First notice that Case~3 always gives a worse bound than Case~2:
the approximation ratio for the former is $1+\min\{r^{-1},a\}$, whereas for the
latter it is $1+a$. Thus we only have to consider the minimum between Cases~1 and 3.

Our goal then is to find a choice of $r$ and $a$ that achieves the largest possible
such value. We can formulate this as a nonlinear optimization problem on the
variables $r$ and $a$. To simplify the exposition, we also consider an auxiliary
variable $\rho$, which will be set to the minimum of the approximation ratios:
\[\rho=\min\left\{1+r+a^{-1}+\ldots+a^{-n+2},1+\min\{r^{-1},a\}\right\}=\min\left\{1+r+a^{-1}+\ldots+a^{-n+2},1+r^{-1},1+a\right\}.\]
This can be enforced by the constraints $\rho\leq1+r+a^{-1}+\ldots+a^{-n+2}$,
$\rho\leq1+r^{-1}$ and $\rho\leq1+a$. Thus, our optimization problem becomes

\begin{align}
\sup\quad & \rho \label{eq:nonlinopt}\tag{NLP}\\
\text{s.t.}\quad &\rho \leq 1+r+a^{-1}+\ldots+a^{-n+2}\nonumber\\
&\rho \leq 1+r^{-1}\nonumber\\
&\rho \leq 1+a\nonumber\\
&0 < r < 1 < a\nonumber\\
&1-r > a^{-1}-a^{-n+2}\nonumber
\end{align}

Notice that \emph{any} feasible solution of \eqref{eq:nonlinopt} gives rise to a
lower bound on the approximation ratio of truthful machine scheduling. In our next
lemma, we characterize the limiting optimal solution of the above optimization
problem. Thus, the lower bound achieved corresponds to the best possible lower bound
using the general construction in this paper.

\begin{lemma}
\label{lem:opt-solution}
An optimal solution to the optimization problem given by~\eqref{eq:nonlinopt} is as follows.

\begin{enumerate}
	\item\label{item:opt-solution-small-n} For $n=3,4,5$, choose $\rho=1+a$, $r=\frac{1}{a}$, and $a$ as the positive solution of the equation
	\begin{align*}
	\frac{2}{a}=a,&\quad\text{for}\;\; n=3;\\
	\frac{2}{a}+\frac{1}{a^2}=a,&\quad\text{for}\;\; n=4;\\
	\frac{2}{a}+\frac{1}{a^2}+\frac{1}{a^3}=a,&\quad\text{for}\;\; n=5.
    \end{align*}
	\item\label{item:opt-solution-large-n} For $n\geq 6$, choose $\rho=1+a$, $r=1-\frac{1}{a}+\frac{1}{a^{n-2}}$, and $a$ as the positive solution of the equation
	\begin{equation}
    \label{eq:opt-solution-large-n}
    1+\frac{1}{a^2}+\cdots+\frac{1}{a^{n-3}}+\frac{2}{a^{n-2}}=a.
    \end{equation}
\end{enumerate}
\end{lemma}

We defer the (admittedly technical) proof of \cref{lem:opt-solution} to \cref{sec:nonlinopt} below; for the time being, we show how this lemma allows us to prove our main result.

\begin{theorem}
\label{th:main-lower-bound}
No deterministic truthful mechanism for unrelated machine scheduling can have an approximation ratio better than $\rho\approx 2.755$, where $\rho$ is the (unique real) solution of equation 
\begin{equation}
\label{eq:lower-bound-infinity}
(\rho-1)(\rho -2)^2 = 1.
\end{equation}
For a restricted number of machines the lower bounds can be seen in~\cref{table:lower-bounds}.
\end{theorem}

\begin{proof}
For $n$ large enough  we can use Case~\ref{item:opt-solution-large-n}
    of~\cref{lem:opt-solution}. In particular, taking the limit
    of~\eqref{eq:opt-solution-large-n} as $n\rightarrow\infty$, we can ensure a
    lower bound of $\rho=a+1$, where $a$ is the (unique) real solution of equation
    $$1+\sum_{i=2}^\infty\frac{1}{a^i} = 1+\frac{1}{a(a-1)} = a.$$ Performing the
    transformation $a=\rho -1$, and multiplying throughout by $(\rho-1)(\rho-2)$, we get exactly~\eqref{eq:lower-bound-infinity}.

For a fixed number of machines $n$, we can directly solve the equations given by
either Case~\ref{item:opt-solution-small-n} ($n=3,4,5$) or
Case~\ref{item:opt-solution-large-n} of~\cref{lem:opt-solution} to derive the
corresponding value of $a$, for a lower bound of $\rho=a+1$. In particular, for
$n=3,4,5$ one gets $a=\sqrt{2}\approx 1.414$, $a=\phi\approx 1.618$ (i.e., the
\emph{golden ratio}) and $a\approx 1.711$, respectively. The values of $\rho$ for up to
$n=8$ machines are given in~\cref{table:lower-bounds}.
\end{proof}

\subsection{Proof of \texorpdfstring{\cref{lem:opt-solution}}{Lemma~3}}\label{sec:nonlinopt}

For the remainder of the paper we focus on proving \cref{lem:opt-solution}, that is,
we characterize the limiting optimal solution of \eqref{eq:nonlinopt}. We begin by
introducing a new variable $z=a^{-1}$, and restate the problem in terms of
$r,z,\rho$.

\begin{align}
\sup\quad & \rho \label{eq:prob2} \\ 
\text{s.t.}\quad &\rho \leq 1+r+z+\ldots+z^{n-2}\nonumber\\
&\rho \leq 1+r^{-1}\nonumber\\
&\rho \leq 1+z^{-1}\nonumber\\
&0 < r , z < 1\nonumber\\
&r < 1 - z + z^{n-2}\nonumber
\end{align}

Notice that the function
$(r,z)\mapsto\min\{1+r+z+\ldots+z^{n-2},1+r^{-1},1+z^{-1}\}$, defined in the
feasibility domain $D=\{(r,z)\;:\;0<r,z<1\text{ and }r<1-z+z^{n-2}\}$, has a
continuous extension to the closure $\bar{D}=
\{(r,z)\;:\;0\leq r,z\leq 1\text{ and }r\leq 1-z+z^{n-2}\}$, which is a compact set.
By the extreme value theorem, the continuous extension must achieve its supremum at
some point in $\bar{D}$; that is to say, the supremum of \eqref{eq:prob2}
corresponds to the maximum of the relaxed problem,

\begin{align}
\max\quad & \rho \nonumber \\ 
\text{s.t.}\quad &\rho \leq 1+r+z+\ldots+z^{n-2}\nonumber\\
&\rho \leq 1+r^{-1}\nonumber\\
&\rho \leq 1+z^{-1}\nonumber\\
&0 \leq r , z \leq 1\nonumber\\
&r \leq 1 - z + z^{n-2}\label{eq:prob3cons5}
\end{align}
which always exist.

Let $(r,z,\rho)$ be an optimal solution. Our next step is to prove that
$\rho=1+z^{-1}$. Suppose otherwise; then, since
$\rho=\min\{1+r+z+\ldots+z^{n-2},1+r^{-1},1+z^{-1}\}$, one must have that either
\begin{equation}\label{eq:onepluszinvnottight}1+r+z+\ldots+z^{n-2}<1+z^{-1}\qquad\text{or}\qquad1+r^{-1}<1+z^{-1}.\end{equation}

We will show that, under such circumstances, we could find a perturbed
$(\tilde{r},\tilde{z},\tilde{\rho})$ with a strictly better objective value, thus
yielding a contradiction. Our analysis proceeds in three cases.

\textbf{Case 1}: $r=0$. This implies that $1+r^{-1}=\infty$, and thus
$\rho=1+r+z+\ldots+z^{n-2}<1+z^{-1}\leq 1+r^{-1}$ Also, since $1-z+z^{n-2}>0$ for
$0\leq z\leq 1$, \eqref{eq:prob3cons5} is not tight, that is to say,
$r<1-z+z^{n-2}$. Thus, we can increase $r$ by an arbitrarily small $\varepsilon>0$,
thus yielding a feasible solution $(r+\varepsilon,z,\rho+\varepsilon)$ with a
strictly better objective value.

\textbf{Case 2}: $r>0$ and $z=1$. This cannot occur, since it would imply both
\[1+r+z+\ldots+z^{n-2}\geq 1+z^{-1}\qquad\text{and}\qquad1+r^{-1}\geq 1+z^{-1},\]
which would contradict \eqref{eq:onepluszinvnottight}.

\textbf{Case 3}: $r>0$ and $z<1$. Take $\varepsilon>0$ sufficiently small and
perturb $(r,z)$ to a new pair $(r-\varepsilon,z+\varepsilon)$, so that
$r-\varepsilon>0$, $z+\varepsilon<1$, and \eqref{eq:onepluszinvnottight} remains
valid. Notice that, under this perturbation, $r$ decreases, $z$ increases, and $r+z$
remains constant. Hence, we do not leave the feasibility region; in particular,
\eqref{eq:prob3cons5} can be written as $r+z\leq 1+z^{n-2}$, and this inequality can
only remain valid after the perturbation. Finally, the perturbation increases both
left-hand sides and decreases both right-hand sides of
\eqref{eq:onepluszinvnottight}. Therefore, the perturbed $(\tilde{r},\tilde{z})$
gives rise to a strictly better objective value.

We have thus deduced that $\rho=1+z^{-1}$ in an optimal solution. This allows us to
restate the optimization problem,
\begin{align*}
\max\quad & 1+z^{-1}\\ 
\text{s.t.}\quad &1+z^{-1} \leq 1+r+z+\ldots+z^{n-2}\\
&1+z^{-1} \leq 1+r^{-1}\\
&0\leq r, z\leq 1\\
&r\leq 1-z+z^{n-2}
\end{align*}
Further rearranging, and removing unnecessary inequalities, yields
\begin{align*}
\max\quad & 1+z^{-1}\\ 
\text{s.t.}\quad &r\geq z^{-1}-z-\ldots-z^{n-2}\\
&r\geq 0\\
&r\leq z\\
&r\leq 1-z+z^{n-2}\\
&0\leq z\leq 1
\end{align*}
Next observe that we can remove the dependency on $r$ by setting
$r=\min\{z,1-z+z^{n-2}\}$, as long as a feasible choice of $r$ exists. Thus, we end
up with
\begin{align}
\max\quad & 1+z^{-1}\nonumber\\ 
\text{s.t.}\quad &z^{-1}-z-\ldots-z^{n-2}\leq z\label{eq:prob6const1}\\
&z^{-1}-z-\ldots-z^{n-2}\leq 1-z+z^{n-2}\label{eq:prob6const2}\\
&0\leq z\leq 1\nonumber\\
&0\leq 1-z+z^{n-2}\label{eq:prob6const4}
\end{align}

Notice that \eqref{eq:prob6const4} is redundant from $0\leq z\leq 1$, and can be
removed. Also, we can rewrite \eqref{eq:prob6const1} and \eqref{eq:prob6const2} as
\[z^{-1}\leq 2z+z^{2}+\ldots+z^{n-3}+z^{n-2},\qquad\qquad z^{-1}\leq
1+z^{2}+\ldots+z^{n-3}+2z^{n-2}.\] In both of the above inequalities, the left hand
side is decreasing in $z$, from $\infty$ as $z\rightarrow 0$ to $1$ at $z=1$,
whereas the right hand side is increasing in $z$, from either $0$ or $1$ at $z=0$ to
$n-1$ at $z=1$. Hence, there are unique positive solutions $z_{n,1}$, $z_{n,2}$ to
the equations
\begin{align*} 
z^{-1}&=2z+z^{2}+\ldots+z^{n-3}+z^{n-2};\\
z^{-1}&=1+z^{2}+\ldots+z^{n-3}+2z^{n-2};
\end{align*}
and moreover, \eqref{eq:prob6const1}, \eqref{eq:prob6const2} are equivalent to
$z\geq z_{n,1}$, $z\geq z_{n,2}$ respectively. Since $z=1$ is a valid feasible
point, we also get that $0<z_{n,1},z_{n,2}<1$. Since our goal is to maximize
$1+z^{-1}$, this is obtained by taking $z$ to be the maximum of $z_{n,1}$,
$z_{n,2}$.

We can finally convert back to $a=z^{-1}$. Since $0<z<1$, $1<a<\infty$. We recover
$r$ via $r=\min\{z,1-z+z^{n-2}\}=\min\{a^{-1},1-a^{-1}+a^{-n+2}\}$. Also, the
reciprocals of $z_{n,1}$, $z_{n,2}$ correspond to the unique positive solutions
$a_{n,1}$, $a_{n,2}$ to the equations
\begin{align}
a&=2a^{-1}+a^{-2}+\ldots+a^{-n+2};\label{eq:an1}\\
a&=1+a^{-2}+\ldots+a^{-n+3}+2a^{-n+2};\label{eq:an2}\end{align}
and the maximum of $z_{n,1}$, $z_{n,2}$ corresponds to the minimum of $a_{n,1}$, $a_{n,2}$.

Now, for $n=3,4,5$, one can numerically check that $a_{n,1}<a_{n,2}$: we have
\begin{center}\begin{tabular}{ccc}
$a_{3,1}\approx 1.414$ & $a_{4,1}\approx 1.618$ & $a_{5,1}\approx 1.711$\\
$a_{3,2}\approx 1.618$ & $a_{4,2}\approx 1.696$ & $a_{5,2}\approx
1.725$\end{tabular}\end{center} Thus, for $n=3,4,5$, the optimal solution
corresponds to taking $a$ such that $a=2a^{-1}+a^{-2}+\ldots+a^{-n+2}$; and
therefore, \eqref{eq:prob6const1} is tight, so that $r=a^{-1}$. On the other hand,
for $n=6$, we have $a_{n,1}\approx1.755>1.739\approx a_{n,2}$; and moreover, as we
increment $n$, the right hand side of \eqref{eq:an1} increases by an extra term
$a^{-n+2}$ whereas the right hand side of \eqref{eq:an2} increases by
$2a^{-n+2}-a^{-n+3}=a^{-n+2}(2-a)$, which is nonnegative: by plugging $a=2$ in
\eqref{eq:an2} we see that $a_{n,2}<2$. Hence, the sequences $a_{n,1}$ and $a_{n,2}$
are both increasing, and in particular $a_{n,2}$ converges to some value
$a_{\infty,2}$ which is the solution of
\[a=1+\sum_{i=2}^\infty\frac{1}{a^i} = 1+\frac{1}{a(a-1)}.\]
We can directly check that $a_{\infty,2}=a_{6,1}\approx 1.755$, by comparing the respective equations:
\begin{align*}
a_{\infty,2}=1+\frac{1}{a_{\infty,2}(a_{\infty,2}-1)}&\Rightarrow
a_{\infty,2}(a_{\infty,2}-1)^2=1\\ &\Rightarrow
a_{\infty,2}^3-2a_{\infty,2}^2+a_{\infty,2}-1=0;\\
a_{6,1}=2a_{6,1}^{-1}+a_{6,1}^{-2}+a_{6,1}^{-3}+a_{6,1}^{-4}&\Rightarrow
a_{6,1}^{5}-2a_{6,1}^{3}-a_{6,1}^{2}-a_{6,1}=1\\ &\Rightarrow
(a_{6,1}^3-2a_{6,1}^2+a_{6,1}-1)(a_{6,1}+1)^2=0.\end{align*} Thus, for $n\geq 6$, we
have that $a_{n,2}<a_{\infty,2}=a_{6,1}\leq a_{n,1}$. We conclude that the optimal
solution corresponds to taking $a$ such that $a=1+a^{-2}+\ldots+a^{-n+3}+2a^{-n+2}$;
this means that \eqref{eq:prob6const2} is tight, so that $r=1-a^{-1}+a^{-n+2}$. This
finishes the proof.

\bibliography{scheduling_new_lower_bound}
\end{document}